\documentclass[12pt,leqno,letterpaper]{article}

\usepackage{amsmath,amsthm,enumerate,amssymb}
\usepackage[utf8]{inputenc}

\usepackage[english]{babel}
\usepackage{graphicx}
\usepackage{color}

\usepackage{graphicx}

\newtheorem{theorem}{Theorem}[section]
\newtheorem{proposition}[theorem]{Proposition}
\newtheorem{lemma}[theorem]{Lemma}
\newtheorem{corollary}[theorem]{Corollary}

\newtheorem{remark}[theorem]{Remark}

\newcommand{\tr}{{\rm Tr\hskip -0.2em}~}

\DeclareMathOperator{\frechetdiff}{\mathit d}
\newcommand{\fd}[1]{\frechetdiff\hskip -0.3em{#1}}
\newcommand{\sfd}[1]{\frechetdiff^2\hskip -0.3em{#1}}

\newcommand{\mean}{\mathbb E\hskip 1pt}

\begin{document}

\title{Multivariate extensions\\of the\\Golden-Thompson inequality}
\author{Frank Hansen}
\date{June 22, 2014}

\maketitle

\begin{abstract} We study concave trace functions of several operator variables and formulate and prove multivariate generalisations of the Golden-Thompson inequality. The obtained results imply that certain functionals in quantum statistical mechanics have bounds of the same form as they appear in classical physics.\\[2ex]
{\bf MSC2010} classification: 47A63.\\[0.5ex]

\end{abstract}

\section{Introduction}

The Golden-Thompson inequality, which is of importance in statistical mechanics and in the theory of random matrices, states that
\[
\tr e^{L+B}\le \tr e^L e^B
\]
for arbitrary self-adjoint matrices $ L $ and $ B. $ It is known that there is no direct extension of this inequality to more operator variables, and there is an extensive literature investigating these matters, cf. \cite{kn:hiai:1993, kn:ando:1994:2, kn:forrester:2014} and the references therein.

We prove, among other statements, the following extension of the Golden-Thompson inequality. Consider $ n\times m $ matrices  $ H_1,\dots,H_k $ with
\[
H_1^*H_1+\cdots+H_k^*H_k= 1_n\,.
\]
Then the inequality
\begin{equation}\label{Main result}
\tr\exp\Bigl(L+\sum_{i=1}^k H^*_i B_i H_i\Bigr)\le\tr\exp\bigl(L) \sum_{i=1}^k H_i^* (\exp B_i) H_i
\end{equation}
is valid for arbitrary self-adjoint $ n\times n $ matrices $ L $ and $ B_1,\dots, B_k\,. $ This is, for $ n=m, $ the same bound as obtained when all the matrices commute. We are thus allowed to estimate partition functions or the Helmhotz function in quantum statistical mechanics and obtain bounds on the same form as they appear in classical physics. 

We obtain, as a simple special case, the inequality
\begin{equation}\label{inequality on simpel form}
\tr \exp\bigl(L+{\textstyle\frac{1}{2}}A+{\textstyle\frac{1}{2}}B\bigr)\le\tr (\exp L)\bigl(\textstyle{\frac{1}{2}}\exp A+\textstyle{\frac{1}{2}} \exp B\bigr)
\end{equation}
valid for arbitrary self-adjoint matrices $ L, $ $ A $ and $ B. $ Notice that (\ref{inequality on simpel form}) reduces to the Golden Thompson inequality for $ A=B $ and to convexity under the trace of the exponential function for $ L=0. $ The inequality may thus be considered as an interpolation inequality between Golden-Thompson's inequality and Jensen's inequality. However, we cannot derive  (\ref{inequality on simpel form}) from these special cases. If we first apply Golden-Thompson's inequality then we obtain
\[
\tr \exp\bigl(L+{\textstyle\frac{1}{2}}A+{\textstyle\frac{1}{2}}B\bigr)\le\tr (\exp L)\exp\bigl(\textstyle{\frac{1}{2}}A+\textstyle{\frac{1}{2}} B\bigr),
\]
but this is insufficient to obtain (\ref{inequality on simpel form}) since  $ L $ is arbitrary and the exponential function is not operator convex.

\section{Preliminaries}

The following lemma is both well-known and very useful. We include the proof for the benefit of the reader.

\begin{lemma}\label{differential inequality}
Let $ \varphi\colon\mathcal D\to\mathcal A_\text{sa} $ be a map defined in a convex cone $ \mathcal D\subseteq X $ of a Banach space $ X $ with values in the self-adjoint part of a $ C^* $-algebra $ \mathcal A. $ If $ \varphi $  is Fréchet differentiable, convex and positively homogeneous then
\[
d\varphi(x)h\le\varphi(h)\qquad x,h\in\mathcal D,
\]
where $ \fd{}\varphi(x) $ denotes the Fréchet differential of $ \varphi(x). $
\end{lemma}

\begin{proof}
Since
\[
\begin{array}{rl}
\varphi(x+th)&=\displaystyle (1+t)\varphi\Bigl(\frac{1}{1+t}x+\frac{t}{1+t}h\Bigr)\\[2ex]
&\displaystyle\le (1+t)\Bigl(\frac{1}{1+t}\varphi(x)+\frac{t}{1+t}\varphi(h)\Bigr)\\[2.5ex]
&=\varphi(x)+t\varphi(h)
\end{array}
\]
for $ 0\le t\le 1, $ we obtain
\[
\frac{\varphi(x+th)-\varphi(x)}{t}\le \varphi(h)\qquad 0<t\le 1
\]
and thus $ \fd{}\varphi(x)h\le\varphi(h). $
\end{proof}

We refer to the monograph \cite{kn:flett:1980} for a general account of Fréchet differentiable mappings between Banach spaces. 

The logarithm is operator monotone with Lebesgue measure as representing measure, thus
\[
\log x=\int_0^\infty\left(\frac{1}{t+1}-\frac{1}{x+t}\right)\,dt\qquad x>0.
\]
Since
\[
x^{1/2}(x+h)^{-1}x^{1/2}=(1+x^{-1/2}hx^{-1/2})^{-1}
=1-x^{-1/2}hx^{-1/2}+o(h),
\]
we derive that $ \fd{} x^{-1}h=-x^{-1}hx^{-1}. $ Notice that $ h $ may be arbitrary as we are not using the functional calculus. Consequently
\[
\fd{}\log(x)h=\displaystyle\int_0^\infty (x+t)^{-1}h(x+t)^{-1}\,dt.
\]
We have thus obtained the following integral expression
\begin{equation}\label{Q(x,h)}
Q(x,h)=\tr h^*\,\fd{}\log(x)h=\tr\int_0^\infty h^* (x+t)^{-1}h(x+t)^{-1}\,dt.
\end{equation}
It follows from the integral expression that $ Q(x,h) $ is positively homogeneous in $ (x,h). $ Lieb proved that it is a convex function in two variables \cite[Theorem 3]{kn:lieb:1973:1}. But this is a reflection of a quite general result. Zhang and the author recently proved \cite{kn:hansen:2014:3} that for a strictly increasing continuously differentiable function $ f\colon(0,\infty)\to\mathbf R $ the form 
\[
(x,h)\to\tr h^*\fd{}f(x)h\qquad x>0
\]
is convex if and only if the derivative of $ f $ is operator convex and numerically decreasing. 

We shall now use 
Lemma~\ref{differential inequality}. Notice that $ Q(x,h) $ is defined in the cone $ \mathcal D=B(H)_+\times B(H), $ where $ H $ is a finite dimensional Hilbert space. Thus
\begin{equation}\label{inequality for dQ}
\fd{}Q(x,h)(y,k)\le Q(y,k)
\end{equation}
for positive definite $ x,y $ and arbitrary $ h,k. $

We end this section by giving a new result for the form $ Q $ that will prove crucial in the rest of the paper.

\begin{proposition}\label{Q-form and a contraction}
Let $ X $ be an invertible contraction. Then
\[
Q(XAX^*,B)\le Q(A,X^{-1}B(X^*)^{-1})
\]
for positive definite $ A $ and arbitrary $ B. $
\end{proposition}

\begin{proof}
We use the integral representation of the form $ Q $ and obtain
\[
Q(XAX^*,B)=\tr\int_0^\infty B^* (XAX^*+t)^{-1}B(XAX^*+t)^{-1}\,dt.
\]
Since $ X $ is a contraction we derive the inequality
\[
\frac{1}{XAX^*+t}\le\frac{1}{X(A+t)X^*}=(X^*)^{-1}(A+t)^{-1}X^{-1}.
\]
Under the trace this inequality implies
\[
\begin{array}{rl}
Q(XAX^*,B)&\le\displaystyle \tr\int_0^\infty B^* (X^*)^{-1}(A+t)^{-1}X^{-1} B (X^*)^{-1}(A+t)^{-1}X^{-1}\,dt\\[3ex]
&=Q(A,X^{-1}B(X^*)^{-1})
\end{array}
\]
which is the desired result.
\end{proof}

\section{Concave trace functions}

\begin{theorem}\label{Main concavity theorem} Let $ H $ be a contraction. Then the trace function
\[
\varphi(A)=\tr\exp\bigl(H^*(\log A) H\bigr)
\]
is concave in positive definite matrices.
\end{theorem}

\begin{proof} We may without loss of generality assume that $ H $ is invertible. We calculate the first Fréchet differential
\[
\begin{array}{rl}
\fd{} \varphi(A)B &=\tr \fd{}\exp\bigl(H^*(\log A) H \bigr) (H^*(\fd{}\log (A)B) H)\\[1.5ex]
&=\tr \exp\bigl(H^*(\log A) H \bigr) (H^*(\fd{}\log (A)B) H),
\end{array}
\]
where we used the identity $ \tr\fd{}f(A)B=\tr f'(A)B $ valid for differentiable functions. We then
consider the following functions of the single operator variable $ A. $ 
\[
\begin{array}{rl}
 C&=H^*(\log A) H \\[1.5ex]
 D&=H^*(\fd{}\log(A)B) H=\fd{}_A(H^*(\log A)H)B=\fd{}_A(C)B \\[1.5ex]
 E&=H\exp(C) H^* = H\exp\bigl(H^*(\log A)H\bigr) H^* \\[1.5ex]
 G&=H \fd{}\exp(C)(D) H^* =\fd{}_C(H\exp(C)H^*)D=\fd{}_C(E)D.
 \end{array}
 \]
For clarity,  we use the notation $ \fd{}_A $ to indicate Fréchet differentiation with respect to $ A $ of compound expressions.
We proceed to calculate the second Fréchet differential
\[
\begin{array}{l}
\sfd{}\varphi(A)(B,B)=\fd{}_A(\fd{}\varphi(A)B)B=\fd{}_A (\tr\exp(C)D)B\\[1.5ex]
=\tr \fd{}\exp (C)( D)D +\tr \exp (C) H^*\sfd{}\log (A)(B,B)H \\[1.5ex]
=\tr \fd{}\exp (C)( D)D +\tr E\, \sfd{}\log(A)(B,B).
\end{array}
\]
We recall \cite{kn:flett:1980} that $ \varphi $ is concave if and only if $ \sfd{}\varphi(A)(B,B)\le 0 $ for positive definite $ A $ and self-adjoint $ B. $
To evaluate the second term we apply the chain rule to the form $ Q(A,B) $ and obtain
\[
\begin{array}{rl}
\fd{}Q(A,B)(a,b)&=\fd{}_1Q(A,B)a+\fd{}_2Q(A,B)b\\[1.5ex]
&=\tr B \sfd{}\log(A)(B,a)+\tr b \fd{}\log(A)B+\tr  B \fd{}\log(A) b
\end{array}
\]
for positive definite $ A,a $ and self-adjoint $ B,b. $ The integral representation in (\ref{Q(x,h)}) implies
$ \tr B \fd{}\log(A)b=\tr b \fd{}\log(A)B, $ and we therefore obtain
\[
\fd{}Q(A,B)(a,b)=\tr B \sfd{}\log(A)(B,a)+2\tr b \fd{}\log(A)B.
\]
By using (\ref{inequality for dQ}) we now obtain the inequality
\[
\begin{array}{rl}
\tr a \sfd{}\log(A)(B,B)&=\tr B \sfd{}\log(A)(B,a)\\[1.5ex]
&=\fd{}Q(A,B)(a,b)-2\tr b \fd{}\log(A)B\\[1.5ex]
&\le Q(a,b)-2\tr b \fd{}\log(A)B
\end{array}
\]
for positive definite $ A, a $ and self-adjoint $ B,b. $
Since $ E $ is positive definite we may put $ a=E $ and thus obtain
\[
\begin{array}{rl}
\sfd{}\varphi(A)(B,B)&=\tr \fd{}\exp (C)( D)D+\tr E\, \sfd{}\log(A)(B,B)\\[1.5ex]
&\le  \tr \fd{}\exp (C)( D)D+Q(E,b)-2\tr b \fd{}\log(A)B.
\end{array}
\]
By setting $ b=G $ we then obtain
\[
\sfd{}\varphi(A)(B,B)\le  \tr \fd{}\exp (C)( D)D+ Q(E,G) -2\tr G \fd{}\log(A)B
\]
for positive definite $ A $ and self-adjoint $ B. $ But since
\[
\begin{array}{rl}
\tr G \fd{}\log(A)B&=\tr H \fd{}\exp(C)(D) H^* \fd{}\log(A)B\\[1.5ex]
&=\tr H^*(\fd{}\log(A)B)H \fd{}\exp(C)D \\[1.5ex]
&=\tr \fd{}(H^*(\log A) H)B \fd{}\exp(C)D\\[1.5ex]
&=\tr(\fd{}(C)B)  \fd{}\exp(C)D\\[1.5ex]
&=\tr D \fd{}\exp(C)D,
\end{array}
\]
we obtain
\[
\sfd{}\varphi(A)(B,B)\le Q(E,G) - \tr D \fd{}\exp(C)D.
\]
We now apply Proposition~\ref{Q-form and a contraction} and obtain
\[
\begin{array}{rl}
Q(E,G)&=Q(H\exp(C)H^*,G)\\[1.5ex]
&\le Q(\exp C, H^{-1}G(H^*)^{-1})\\[1.5ex]
&=Q(\exp C, \fd{}\exp(C)(D))\\[1.5ex]
&=\tr \fd{}\exp(C)(D) \fd{}\log(\exp C)\fd{}\exp(C)(D).
\end{array}
\]
However, since the inverse of the linear map $ h\to\fd{}\exp(x)h $ is given by
\[
\fd{}\exp(x)^{-1}=\fd{}\log(\exp x),
\]
we realise that
\[
\fd{}\log(\exp C)\fd{}\exp(C)(D)=\fd{}\exp(C)^{-1}(\fd{}\exp(C) D)=D.
\]
Therefore,
\[
Q(E,G)\le \tr D \fd{}\exp(C) D
\]
and thus $ \sfd{}\varphi(A)(B,B)\le 0 $ for positive definite $ A $ and self-adjoint $ B. $ This shows that $ \varphi $ is concave.
\end{proof}

\begin{corollary}\label{Concavity corollary 1}
Consider $ n\times m $ matrices  $ H_1,\dots,H_k $ with
\[
H_1^*H_1+\cdots+H_k^*H_k\le 1_n
\]
where $ 1_n $ denotes the $ n\times n $ unit matrix. Then the trace function
\[
\varphi(A_1,\dots,A_k)=\tr\exp\bigl(H^*_1 (\log A_1) H_1+\cdots+ H^*_k (\log A_k) H_k\bigr)
\]
is concave in $ k $-tuples of positive definite $ n\times n $ matrices.  
\end{corollary}

\begin{proof}
We set
\[
A=\begin{pmatrix}
     A_1     & 0     & \cdots   & 0\\
     0         & A_2 &             & 0\\
     \vdots &         & \ddots  & \vdots\\
     0        & 0       & \dots    & A_k
     \end{pmatrix}\qquad\text{and}\quad
H=\begin{pmatrix}
     H_1    & 0        & \cdots & 0\\
     H_2    & 0        & \cdots & 0\\
     \vdots & \vdots & \ddots & \vdots\\
     H_k    & 0         & \cdots & 0
     \end{pmatrix}
\]
with zero matrices of suitable orders inserted and notice that $ H $ is a contraction. Furthermore,
\[
H^*( \log A) H=\begin{pmatrix}
     \sum_{i=1}^k H_i^* (\log A_i) H_i    & 0       & \cdots   & 0\\
     0                                                     & 0        &  \cdots  & 0\\
     \vdots                                             & \vdots & \ddots   & \vdots\\
     0                                                     & 0        & \dots     & 0
     \end{pmatrix}.
\]
Thus
\[
\tr\exp\bigl(H^*(\log A) H\bigr)=\tr\exp\Bigl(\sum_{i=1}^k H^*_i(\log A_i) H_i\Bigr) +(k-1)n
\]
and the statement follows from Theorem~\ref{Main concavity theorem}.
\end{proof}

\begin{corollary}\label{Concavity corollary 2}
Consider $ n\times m $ matrices  $ H_1,\dots,H_k $ with
\[
H_1^*H_1+\cdots+H_k^*H_k\le 1_n
\]
and a self-adjoint $ n\times n $ matrix $ L. $ Then the trace function
\[
\varphi(A_1,\dots,A_k)=\tr\exp\Bigl(L+\sum_{i=1}^k H^*_i (\log A_i) H_i\Bigr)
\]
is concave in $ k $-tuples of positive definite $ n\times n $ matrices.  
\end{corollary}

\begin{proof}
By appealing to continuity we may without loss of generality assume 
\[
H_1^*H_1+\cdots+H_k^*H_k< 1_n
\]
and set $ H_{k+1}=\bigl(1_n-(H_1^*H_1+\cdots+H_k^*H_k)\bigr)^{1/2}. $ Then $ H_{k+1} $ is positive definite and since
\[
H_1^*H_1+\cdots+H_k^*H_k+H^2_{k+1}= 1_n
\]
we deduce from Corollary~\ref{Concavity corollary 1} that the trace function
\[
\begin{array}{l}
\varphi(A_1,\dots,A_k,A_{k+1})\\[1.5ex]
=\tr\exp\bigl(H^*_1 (\log A_1) H_1+\cdots+ H^*_k (\log A_k) H_k+H_{k+1} (\log A_{k+1})H_{k+1}\bigr)
\end{array}
\]
is concave in positive definite matrices. We keep $ A_{k+1} $ constant by setting
\[
A_{k+1}=\exp\bigl(H^{-1}_{k+1} L H_{k+1}^{-1}\bigr)
\]
and the statement now follows.
\end{proof}

\begin{remark}
Corollary~\ref{Concavity corollary 2} contains two celebrated theorems of Lieb. If we set $ k=1 $ and $ H=1 $ then the trace function
\[
\varphi(A)=\tr\exp(L+\log A)
\]
is concave in positive definite matrices, cf. \cite[Theorem 6]{kn:lieb:1973:1}. If $ H_1,\dots,H_k $ are chosen as square roots of positive numbers times the identity matrix then we obtain that the trace function
\[
\varphi(A_1,\dots,A_k)=\tr\exp(L+p_1 \log A_1+\cdots+p_k \log A_k),
\]
defined in positive definite matrices, is concave, where $ p_1,\dots,p_k $ are non-negative numbers with $ p_1+\cdots+p_k\le 1, $ cf. \cite[Corollary 6.1 (1)]{kn:lieb:1973:1}. 
\end{remark}

\begin{corollary}
Let $ L $ be a fixed self-adjoint matrix, and let $ A_1,\dots,A_k $ be random self-adjoint matrices. Then the inequality
\[
\mean\tr\exp\Bigl(L+\sum_{i=1}^k H_i^* A_i H_i\Bigr)\le \tr\exp\Bigl(L+\sum_{i=1}^k H_i^*(\log\mean e^{A_i}) H_i\Bigr)
\]
holds 
for fixed matrices $ H_1,\dots,H_k $ with $ H_1^*H_1+\cdots+H_k^*H_k\le 1. $
\end{corollary}

The result follows directly from Corollary~\ref{Concavity corollary 2} by applying Jensen's inequality, cf. also \cite[Corollary 3.3]{kn:Tropp:2012}. A simple consequence is that
\[
\mean\tr\exp\Bigl(L+\frac{A_1+\cdots+A_k}{k}\Bigr)\le \tr\exp\Bigl(L+\frac{\log\mean \displaystyle e^{A_1}+\cdots+ \log\mean  e^{A_k}}{k}\Bigr)
\]
for a fixed self-adjoint matrix $ L $ and random self-adjoint matrices $ A_1,\dots,A_k. $ 

\section{Multivariate trace inequalities}

\begin{lemma}\label{Lemma to prove Golden-Thomson type inequalities} 
Consider $ n\times m $ matrices  $ H_1,\dots,H_k $ with
\[
H_1^*H_1+\cdots+H_k^*H_k= 1_n
\]
and a self-adjoint $ n\times n $ matrix $ L. $ Then we have the inequality
\[
\begin{array}{l}
\displaystyle\tr\exp\Bigl(L+\sum_{j=1}^k H^*_j (\log B_j) H_j\Bigr)\\[1.5ex]
\le\displaystyle\tr\exp\Bigl(L+\sum_{j=1}^k H_j^*(\log A_j)H_j\Bigr)\sum_{i=1}^k H_i^*(\fd{}\log(A_i)B_i) H_i
\end{array}
\]
for positive definite $ n\times n $ matrices $ A_1,\dots,A_k $ and $ B_1,\dots,B_k\,. $
\end{lemma}

\begin{proof}
Since the trace function
\[
\varphi(A_1,\dots,A_k)=\tr\exp\Bigl(L+\sum_{j=1}^k H_j^*(\log A_j)H_j \Bigr)
\]
is concave and (positively) homogeneous, we may apply Lemma~\ref{differential inequality} and obtain the inequality
\[
d\varphi(A_1,\dots,A_k)(B_1,\dots,B_k)\ge\varphi(B_1,\dots,B_k).
\]
By applying the chain rule for Fréchet differentials we then derive
\[
\begin{array}{l}
\varphi(B_1,\dots,B_k)\le\displaystyle\ \sum_{i=1}^k d_i\varphi(A_1,\dots,A_k)B_i\\[1.5ex]
=\displaystyle\sum_{i=1}^k\tr\fd{}\exp\Bigl(L+\sum_{j=1}^k H_j^*(\log A_j)H_j \Bigr)H_i^*(\fd{}\log(A_i)B_i) H_i\\[2.5ex]
=\displaystyle\sum_{i=1}^k\tr\exp\Bigl(L+\sum_{j=1}^k H_j^*(\log A_j)H_j \Bigr)H_i^*(\fd{}\log(A_i)B_i) H_i
\end{array}
\]
and the statement follows.
\end{proof}

\begin{theorem}\label{My generalisation of the Golden-Thompson inequality}
Consider $ n\times m $ matrices  $ H_1,\dots,H_k $ with
\[
H_1^*H_1+\cdots+H_k^*H_k= 1_n\,.
\]
Then we have the inequality
\[
\tr\exp\Bigl(L+\sum_{i=1}^k H^*_i B_i H_i \Bigr)\le\tr\exp\bigl(L) \sum_{i=1}^k H_i^* (\exp B_i) H_i
\]
for arbitrary self-adjoint $ n\times n $ matrices $ L $ and $ B_1,\dots, B_k\,. $
\end{theorem}

\begin{proof}
Choose in Lemma~\ref{Lemma to prove Golden-Thomson type inequalities} for $ i=1,\dots,k $  the matrix $ A_i $ as the identity matrix. Then the Fréchet differential $ \fd{}\log(A_i)B_i=B_i $ and $ \log A_i=0. $ The result then follows by replacing $ B_i $ with $ \exp B_i $ for $ i=1,\dots,k. $
\end{proof}

The above inequality is a direct generalisation of the Golden-Thompson inequality. Indeed, if we put $ k=1 $ and take $ H_1 $ as the identity matrix then the inequality in Theorem~\ref{My generalisation of the Golden-Thompson inequality} takes the form
\[
\tr e^{L+B}\le \tr e^L e^B,
\]
cf. \cite{kn:golden:1965, kn:thompson:1965, kn:lieb:1973:1}. We may obtain other corollaries of Lemma~\ref{Lemma to prove Golden-Thomson type inequalities} .

\begin{theorem}\label{Generalised Golden-Thompson inequality}
Consider $ n\times m $ matrices  $ H_1,\dots,H_k $ with
\[
H_1^*H_1+\cdots+H_k^*H_k= 1_n\,.
\]
Then we have the inequality
\[
\tr\exp\Bigl(\sum_{i=1}^k H^*_i (\log B_i-\log A_i) H_i \Bigr)\le\sum_{i=1}^k\tr H_i^*(\fd{}\log(A_i)B_i) H_i
\]
for positive definite $ n\times n $ matrices $ A_1,\dots,A_k $ and $ B_1,\dots,B_k\,. $
\end{theorem}

\begin{proof} The result follows by setting
\[
L=-\bigl(H_1^*(\log A_1) H_1+\cdots+H_k^*(\log A_k) H_k\bigr)
\]
in Lemma~\ref{Lemma to prove Golden-Thomson type inequalities}.
\end{proof}

If we in Theorem~\ref{Generalised Golden-Thompson inequality} put $ k=1 $ and take $ H_1 $ as the unit matrix we obtain
\[
\tr\exp(\log B-\log A)\le \tr\fd{}\log(A) B=\tr A^{-1}B
\]
which is the Golden-Thompson inequality.\\[1ex] 
Furthermore, if $ A_i $ and $ B_i $ commute for $ i=1,\dots,k $ then 
Theorem~\ref{Generalised Golden-Thompson inequality} reduces to the inequality
\[
\tr\exp\Bigl(\sum_{i=1}^k H^*_i (\log B_i-\log A_i) H_i \Bigr)\le\sum_{i=1}^k\tr H_i^* B_i A_i^{-1}  H_i
\]
which is an expression of operator convexity of the exponential function under the trace.

\begin{theorem}\label{Generalised extended Golden-Thompson inequality}
Consider $ n\times m $ matrices  $ H_1,\dots,H_k $ with
\[
H_1^*H_1+\cdots+H_k^*H_k= 1_n\,.
\]
Then we have the inequality
\[
\begin{array}{l}
\displaystyle\tr\exp\Bigl(\sum_{i=1}^k H^*_i (\log B_i+\log C_i-\log A_i) H_i \Bigr)\\[1.5ex]
\le\displaystyle\tr\exp\Bigl(\sum_{i=1}^k H_i^*(\log C_i)H_i \Bigr) \sum_{i=1}^k H_i^*(\fd{}\log(A_i)B_i) H_i
\end{array}
\]
for positive definite $ n\times n $ matrices $ A_1,\dots,A_k, B_1,\dots,B_k $ and $ C_1,\dots,C_k\,. $
\end{theorem}

\begin{proof}
The result follows by setting 
\[
L=H_1^*(\log C_1-\log A_1)H_1+\cdots+H_k^*(\log C_k-\log A_k)H_k 
\]
in Lemma~\ref{Lemma to prove Golden-Thomson type inequalities}.
\end{proof}

If we in Theorem~\ref{Generalised extended Golden-Thompson inequality} put $ k=1 $ and take $ H_1 $ as the unit matrix we
obtain
\[
\tr\exp(\log B+\log C-\log A)\le \tr C\, d\log(A)B,
\]
which is the extended Golden-Thompson inequality. The extended Golden-Thompson inequality reduces to the Golden-Thompson inequality if $ A $ and $ B $ commute.

{\small
 

\vfill

\noindent Frank Hansen: Institute for Excellence in Higher Education, Tohoku University, Japan.\\
Email: frank.hansen@m.tohoku.ac.jp.
      }

\end{document}